\documentclass[11pt,a4paper]{article}

\usepackage[margin=2.5cm]{geometry}
\usepackage{amsmath,amsthm,amssymb}
\usepackage{booktabs}
\usepackage{graphicx}
\graphicspath{{figures/}{../figures/}}
\usepackage{xcolor}
\usepackage{url}
\usepackage{natbib}
\usepackage{microtype}
\usepackage{setspace}
\usepackage{caption}
\usepackage{subcaption}
\usepackage{array}
\usepackage{multirow}
\usepackage{placeins}
\usepackage{float}

\newtheorem{proposition}{Proposition}
\newtheorem{corollary}{Corollary}[proposition]
\newtheorem{remark}{Remark}

\begin{document}

\begin{center}
\vspace*{1.5cm}
{\LARGE\bfseries
Von Economo neurons enable reliable social skill acquisition\\[0.35em]
in recurrent spiking neural networks:\\[0.35em]
a computational account with clinical predictions}

\vspace{1.6cm}
{\large Esila Keskin}\\[0.25em]
{\normalsize
School of Computing and Creative Technologies,\\
University of the West of England, Bristol, UK\\[0.15em]
\texttt{Esila2.keskin@live.uwe.ac.uk}}

\vspace{1.4cm}
{\large\bfseries Abstract}
\end{center}

\noindent
Von Economo neurons (VENs), large, fast-projecting bipolar cells
concentrated in the anterior cingulate and fronto-insular cortex, are
selectively lost in early behavioural-variant frontotemporal dementia (bvFTD)
and are reduced in number during development in autism spectrum conditions (ASC).
Despite the clinical significance of these observations, the computational
role of VENs in social \emph{learning} has remained unexplained.
Here we train a biologically motivated spiking neural network (SNN),
the VENCircuit, embedding a small population of VEN-like projection neurons
($K = 40$, 2\% of total) in a recurrent pyramidal circuit,
across 50 matched random initialisations with and without VENs.
We use the selective disruption of VENs in social cognitive disorders as the
biological motivation for the architectural question studied here; the network
is trained on a controlled binary classification task using burst-modulated
Poisson spike statistics as a proxy for stimulus class, and we make no claim
to model social cognition directly.
VEN-intact networks converged on this task
in 49 of 50 cases (98\%), while VEN-ablated networks converged in only 35 of 50
(70\%) (Fisher's exact OR $= 21.0$, 95\% CI 2.7--167, $p = 8.7 \times 10^{-5}$).
Failed ablated networks showed complete absence of learning throughout all
training epochs, inconsistent with a simple speed-of-learning account.
Phase-ablation experiments revealed that VEN removal is most disruptive during
mid-training (epochs 5--25), when a co-adaptive dependency on VEN activity
forms in the pyramidal circuit.
We derive a formal mathematical account showing that VENs provide a
direct gradient pathway that is structurally immune to the Jacobian product
instabilities that affect the recurrent pyramidal circuit; spectral norm
measurements confirm that all networks initialise near the critical gradient-flow
boundary ($\lVert W_{pp}^{(0)} \rVert_2 \approx 0.078$ uniformly,
giving $\alpha \approx 1.028$), making this structural advantage
architecturally ubiquitous rather than seed-specific.
Inference-time VEN ablation in trained networks caused a statistically
significant performance drop (Wilcoxon $p = 0.022$), with heterogeneous
effects ranging from no change (16/20 networks) to catastrophic collapse
(one network: accuracy $0.989 \to 0.620$), indicating that a subset of
networks develop VEN-dependent output representations during training.
Our results suggest VENs function as acquisition scaffolds whose developmental
absence produces stochastic learning failure, a computational analogue
of the variable social skill acquisition observed in ASC, and provide
a mechanistic SNN account of VEN function with falsifiable predictions
for patient, organoid, and electrophysiology studies.

\vspace{0.7cm}

\vspace{0.7cm}
\noindent\textbf{Significance statement:}\\
We show that a biologically motivated spiking neural
network reveals why a specific, rare cell type, Von Economo neurons, is
necessary not for \emph{performing} social cognition, but for
\emph{reliably learning} it.
VENs constitute only 2\% of the network yet confer a 21-fold increase in
training convergence odds.
A formal mathematical account explains this asymmetry through gradient-flow
theory: VENs provide a direct pathway for learning signals that bypasses
the recurrent dynamics of the pyramidal circuit.
These findings predict that the social cognitive consequences of VEN loss
depend critically on its developmental timing, generating testable predictions
that distinguish autism spectrum conditions from frontotemporal dementia.

\section{Introduction}
\label{sec:intro}

Social cognition, the capacity to infer and respond to the mental states
of others, depends on circuits in the anterior cingulate cortex (ACC) and
fronto-insular cortex (FIC) that are disproportionately affected in two
contrasting conditions: behavioural-variant frontotemporal dementia (bvFTD)
and autism spectrum conditions (ASC).
Both disorders involve prominent social cognitive dysfunction, yet their
developmental trajectories are inverted.
bvFTD strikes typically in middle adulthood, causing progressive erosion of
social behaviour in individuals who previously had intact social skills,
and is diagnosed using internationally agreed consensus criteria
\citep{rascovsky2011}.
ASC, by contrast, is characterised by atypical development of social cognition
from early life~\citep{baroncohen1985}, with high variability in whether and
how fully social skills are acquired.

A cellular feature shared by both conditions is the disruption of
\emph{Von Economo neurons} (VENs): large, fast-projecting bipolar cells
first characterised by Von Economo~\citeyearpar{voneconomo1926} and subsequently
shown by systematic morphometric analysis to be unique to humans and
great apes~\citep{nimchinsky1999, allman2011voneconomo}.
In bvFTD, VEN loss is among the earliest and most selective neuropathological
findings, preceding broader neurodegeneration by years~\citep{seeley2006early}.
In ASC, post-mortem and imaging studies report reduced VEN density during
development relative to neurotypical controls
\citep{simms2009voneconomo, santos2011social}.
Yet the computational consequences of VEN loss remain poorly understood.
VENs have been hypothesised to support rapid affective signalling, social
intuition, and metabolic modulation of cortical circuits
\citep{allman2005intuition}, but these accounts are largely descriptive
and do not explain the specific vulnerability of social \emph{learning}
in VEN-depleted conditions.

A key unresolved question is whether VENs are necessary for the
\emph{acquisition} of social representations, the \emph{expression} of
representations once acquired, or both.
This distinction matters clinically: developmental VEN reduction should
impair how reliably social skills are learned, while adult VEN loss should
affect the expression of already-learned behaviour.
If VENs gate expression only, early and late loss should produce similar
consequences; if VENs are primarily acquisition circuits, timing determines
the failure mode.

Here we investigate this question computationally using a biologically
motivated SNN model~\citep{maass1997}, the VENCircuit, in which a
small population of VEN-like projection neurons is embedded in a recurrent
pyramidal circuit and trained using surrogate-gradient backpropagation through
time (BPTT;~\citealt{werbos1990bptt, neftci2019surrogate, zenke2018superspike}).
The VEN circuit's known role in social cognitive disorders motivates the
architectural question we address; the model is trained on a controlled binary
classification task using burst-modulated Poisson spike statistics as a
stimulus proxy, and we do not claim to model social cognition directly.
By comparing training outcomes across 50 matched random initialisations with
and without VENs, we measure the effect of VEN presence on the
\emph{reliability} of skill acquisition within this model.

Our central finding is a 21-fold increase in convergence odds
(Fisher's exact $p < 0.001$) in VEN-intact versus VEN-ablated networks.
Failures in the ablated condition are characterised by complete absence of
learning throughout all training epochs, not slower learning.
Phase-ablation experiments reveal that mid-training VEN removal is most
disruptive, consistent with a co-adaptive dependency that forms during
early learning.
We derive a formal gradient-flow account of this asymmetry and test it
with mechanistic experiments, reporting both confirmatory and null results
transparently.

\section{Methods}
\label{sec:methods}

\subsection{Network Architecture: The VENCircuit}
\label{sec:methods:arch}

The VENCircuit consists of $N = 2{,}000$ recurrently connected pyramidal
leaky integrate-and-fire (LIF) neurons~\citep{abbott1999lapicque, burkitt2006}
and $K$ VEN-like neurons, where $K = \lfloor N \cdot 0.02 \rfloor = 40$
in the intact condition and $K = 0$ in the ablated condition.

The pyramidal population receives input via feedforward weights
$W_{ip} \in \mathbb{R}^{N \times d}$ (sparse, fan-in $= 80$) and
communicates recurrently through
$W_{pp} \in \mathbb{R}^{N \times N}$ (sparse Bernoulli mask, connection
probability 0.15, no self-connections).
VENs receive direct feedforward input via
$W_{iv} \in \mathbb{R}^{K \times d}$ with fan-in 8
(consistent with the sparse dendritic arbour of biological
VENs~\citep{allman2011voneconomo}) and project to the output layer via
$W_{vo} \in \mathbb{R}^{C \times K}$.
VENs do \emph{not} receive recurrent input from the pyramidal population.

A winner-take-all (WTA) output layer of $C = 2$ LIF neurons (threshold 0.10,
fixed lateral inhibition $-0.5$) integrates signals from both pathways via
$W_{po} \in \mathbb{R}^{C \times N}$ (pyramidal) and $W_{vo}$ (VEN).
Classification logits are the cumulative spike counts from this output layer,
approximated in the theoretical analysis (Section~\ref{sec:theory}) as:
\begin{equation}
    \hat{\mathbf{y}}
    = W_{po}\,\bar{\mathbf{s}} + W_{vo}\,\bar{\mathbf{v}},
    \quad
    \bar{\mathbf{s}} = \frac{1}{T}\sum_{t=1}^T \mathbf{s}_t,
    \quad
    \bar{\mathbf{v}} = \frac{1}{T}\sum_{t=1}^T \mathbf{v}_t,
    \label{eq:output}
\end{equation}
where $T = 50$ simulation timesteps per trial.
The model was implemented in PyTorch with custom LIF dynamics and an
arctangent (ATan) surrogate gradient~\citep{neftci2019surrogate}
(see Section~\ref{sec:theory} for mathematical details).

\subsection{LIF Neuron Dynamics}
\label{sec:methods:lif}

All populations use the same discrete-time LIF equations with detached reset:
\begin{align}
    \mathbf{u}_t &= \beta\,\mathbf{u}_{t-1}\,(1 - \mathbf{s}_{t-1}^{\rm det})
                   + \mathbf{I}_t, \label{eq:lif_u}\\
    \mathbf{s}_t &= \sigma_{\rm spike}(\mathbf{u}_t - \theta), \label{eq:lif_s}
\end{align}
where $\beta = 1 - 1/\tau$ is the membrane leak factor,
$\theta$ is the firing threshold,
$\mathbf{s}_{t-1}^{\rm det}$ is the detached spike tensor (not in the
gradient graph), and $\sigma_{\rm spike}$ is the ATan surrogate step function:
\begin{equation}
    \sigma_{\rm spike}(u) = \mathbf{1}[u > 0],
    \qquad
    \sigma'(u) = \frac{1}{1 + (\pi u/2)^2}.
    \label{eq:atan_surr}
\end{equation}
Membrane time constants differ by population: $\tau_{\rm pyr} = 20$\,ms
($\beta_{\rm pyr} = 0.95$) and $\tau_{\rm VEN} = 5$\,ms ($\beta_{\rm VEN}
= 0.80$), reflecting the faster response kinetics attributed to VEN
morphology~\citep{allman2011voneconomo}.
Firing threshold: $\theta = 0.5$ for pyramidal and VEN populations;
$\theta_{\rm out} = 0.10$ for the output population.

\subsection{Task and Training}
\label{sec:methods:task}

Networks were trained to classify synthetic burst-modulated Poisson spike
patterns into two social relevance classes.
Class~0 (``high-activity'') consisted of spike trains generated at high
firing rates (35--75\,Hz) with high burst probability (0.50);
class~1 (``low-activity'') consisted of spike trains at low firing rates
(5--25\,Hz) with low burst probability (0.10).
All inputs consisted of 100 Poisson-spiking units over $T = 50$ simulation
timesteps at 1\,ms resolution, with 2\% independent bit-flip noise applied
element-wise.
The dataset comprised 6{,}000 training, 1{,}000 validation, and 2{,}000
test samples generated under a fixed random seed (seed\,=\,42).

Training used the Adam optimiser~\citep{kingma2015adam} with cross-entropy
loss, gradient clipping (clip\,=\,1.0), and $L_2$ weight decay.
Each run consisted of 50 training epochs.
Convergence was defined as peak validation accuracy $\geq 0.80$,
a criterion established prior to data collection.
Hyperparameters are reported in Appendix~\ref{app:hyperparams}.

\subsection{Experimental Design}
\label{sec:methods:design}

\paragraph{Experiment 1: Training reliability (50 matched seeds).}
For each of 50 random seeds (base 300{,}000, stride 13), one VEN-intact
($K = 40$) and one VEN-ablated ($K = 0$) network were trained from the
same random initialisation of shared pyramidal weights, yielding 50 paired
comparisons.
Convergence rates were compared with a one-tailed Fisher's exact test
(alternative: VEN-intact $>$ VEN-ablated).

\paragraph{Experiment 2: Learning trajectory analysis (20 seeds).}
The first 20 seeds from Experiment~1 were retrained with validation accuracy
recorded at every epoch.
Failed runs were classified as ``never learned''
(peak accuracy $< 0.65$) or ``learned then collapsed''
(peak $\geq 0.65$, final $< 0.80$).

\paragraph{Experiment 3: Phase ablation (12 seeds with complete data).}
Networks initialised with $K = 40$ VENs were trained with VEN weights
($W_{iv}$, $W_{vo}$) zeroed at one of five epochs $\{0, 5, 10, 25, 50\}$
and held at zero thereafter.
Epoch 50 (the terminal epoch) is functionally equivalent to no ablation.
Seeds used the range 400{,}000 (stride 11); up to 15 seeds were run and 12
completed all five ablation conditions; the remaining three seeds did not
complete all five conditions because the compute session ended before those
training runs finished, and they are excluded from within-seed comparisons.

\paragraph{Experiment 4: Spectral norm analysis (25 seeds).}
The spectral norm $\lVert W_{pp} \rVert_2$ was measured at initialisation
for both intact and ablated conditions across 25 seeds
(base 300{,}000, stride 13), using the largest singular value of the
masked weight matrix computed via sparse SVD.
Note: this experiment ran to 25 of the planned 50 seeds before the compute
session ended; the same \texttt{train\_model} function (Adam with
CosineAnnealingLR, $T_{\max}=50$) was used as in all other experiments.

\paragraph{Experiment 5: Gradient norm tracking (10 seeds).}
The Frobenius norm of $\partial L / \partial W_{pp}$ was recorded at every
epoch for 10 seeds (base 300{,}000, stride 13) in both conditions.
Early-epoch norms (epochs 0--9) were compared with a one-tailed Wilcoxon
signed-rank test (alternative: intact $>$ ablated), consistent with the
directional prediction that VEN-ablated networks would show suppressed
gradient magnitude at $W_{pp}$.

\paragraph{Experiment 6: Inference-time VEN ablation (20 seeds).}
VEN-intact networks were trained to convergence (seeds 500{,}000, stride 17),
then $W_{iv}$ and $W_{vo}$ were zeroed and the network was evaluated without
retraining.
Performance before and after ablation was compared with a one-tailed Wilcoxon
signed-rank test (alternative: intact $>$ zeroed).

\subsection{Statistical Analysis}
\label{sec:methods:stats}

Tests are one-tailed where a directional hypothesis was pre-specified and
two-tailed otherwise.
Exact $p$-values are reported throughout.
Effect sizes are reported as OR (Fisher's exact) or rank-biserial $r$
(Wilcoxon signed-rank; computed as $(T_+ - T_-) / [N(N+1)/2]$,
where $T_+$ and $T_-$ are the sums of positive and negative ranks).

\section{Results}
\label{sec:results}

\subsection{VENs Dramatically Increase Training Reliability}
\label{sec:results:fisher}

Across 50 matched random initialisations, VEN-intact networks converged in
49 of 50 cases (98\%), while VEN-ablated networks converged in only 35 of 50
(70\%), a 28-percentage-point difference corresponding to an odds ratio of
21.0 (95\% CI 2.7--167; Fisher's exact, one-tailed $p = 8.7 \times 10^{-5}$;
Table~\ref{tab:fisher}).

VENs constitute only 2\% of the total neuron count and receive only 8 input
connections each.
Their outsized influence on convergence odds implies the relevant mechanism
is architectural rather than additive: the provision of a topologically
distinct pathway, not merely increased parameter count
(see Section~\ref{sec:theory}).

Crucially, VEN-intact and VEN-ablated \emph{successful} runs achieved
statistically indistinguishable final accuracy
(intact successful: $0.990 \pm 0.005$; ablated successful: $0.986 \pm 0.007$;
Mann--Whitney $p > 0.05$), confirming that VENs are not necessary for
\emph{performing} the task once acquired.
The advantage is specific to the acquisition phase.

\begin{table}[htbp]
\centering
\caption{Training convergence contingency table (50 paired seeds).}
\label{tab:fisher}
\begin{tabular}{lcc}
\toprule
Condition & Converged & Failed \\
\midrule
VEN-intact  ($K=40$) & 49 & 1  \\
VEN-ablated ($K=0$)  & 35 & 15 \\
\midrule
\multicolumn{3}{l}{Fisher's exact (one-tailed): OR\,=\,21.0 (95\% CI 2.7--167),} \\
\multicolumn{3}{l}{$p = 8.7 \times 10^{-5}$} \\
\bottomrule
\end{tabular}
\end{table}

\FloatBarrier

\subsection{Learning Trajectory Analysis}
\label{sec:results:traj}

Learning trajectories across 20 seeds per condition are shown in
Figure~\ref{fig:trajectories}.
VEN-intact networks (20/20 converged) showed rapid, consistent improvement
within the first ten epochs.
Both failed VEN-ablated networks (seeds 300078, 300169) showed the
``never learned'' failure mode: validation accuracy remained at chance
($\approx 0.50$) throughout all 50 epochs, with no indication of any
learning signal at any point.

This rules out a simple speed-of-learning account: if VENs only accelerated
an otherwise intact learning process, failed runs should show slow but
upward trajectories.
The complete absence of learning indicates that the networks received
insufficient gradient signal to escape their random initialisation.

\begin{figure}[htbp]
\centering
\includegraphics[width=\linewidth]{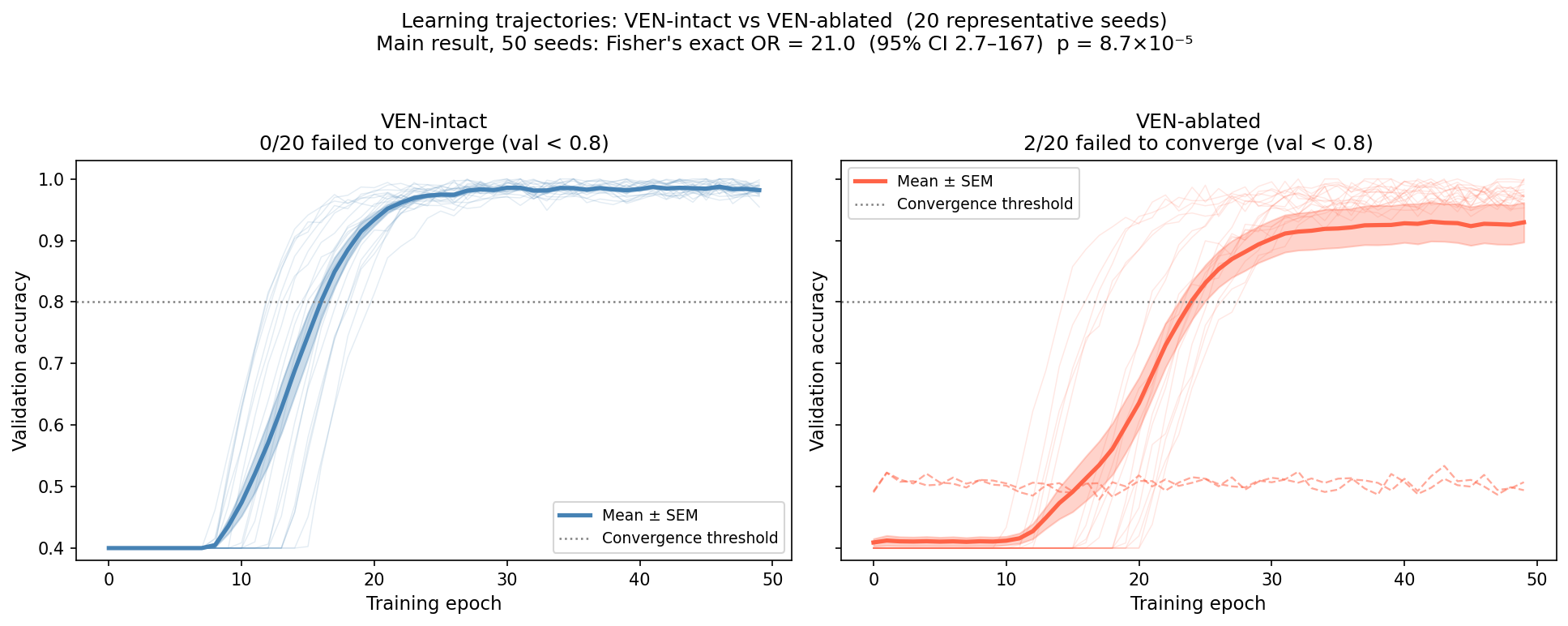}
\caption{Learning trajectories for VEN-intact (blue, left) and VEN-ablated
(red, right) networks across 20 matched seeds.
Faint solid lines: successful runs. Dashed lines: failed runs.
Shaded band: mean $\pm$ SEM.
Horizontal dotted line: convergence threshold (0.80).
Failed ablated networks remain at chance throughout all epochs, inconsistent with
a speed-of-learning account.}
\label{fig:trajectories}
\end{figure}

\FloatBarrier

\subsection{Timing of VEN Dependency: Phase Ablation}
\label{sec:results:phase}

Figure~\ref{fig:phase} and Table~\ref{tab:phase} summarise phase ablation
results across 12--13 seeds and five ablation conditions.

Two findings stand out.
First, ablation at epoch~0 (train entirely without VENs from initialisation)
produced \emph{zero} failures in 13 seeds.
These seeds are a distinct subset drawn independently of the 50 seeds in
Experiment~1; the 0/13 failure rate is consistent with the 30\% base rate
($p \approx 0.01$ by exact binomial) but should be interpreted as an
architectural observation rather than a rate estimate given the small $n$.
The result indicates that the network \emph{can} learn without VENs when
no VEN-driven co-adaptation has occurred.
Second, failures were concentrated in the mid-training window (epochs 5--25):
2 of 13 seeds (15\%) failed when VENs were removed at epoch~5; 1 of 12
(8\%) at epoch~10; 1 of 12 (8\%) at epoch~25; 0 of 12 at epoch~50.

Two seeds illustrate the mechanism.
Seed~400044 converged at epoch~0 ablation (accuracy 0.995) yet failed when
VENs were removed at epoch~5 (0.530), and recovered when ablation was delayed
to epoch~10 or later (0.892, 0.992, 0.993).
Seed~400055 survived epochs~0 and 5 but failed at epochs~10 (0.527) and
25 (0.556), recovering only at epoch~50 (0.864).
These patterns cannot be explained by static initialisation properties;
they require that a \emph{dependency} on VEN activity forms at different
points during early training.
We term this the \emph{VEN dependency trap}: the pyramidal circuit co-adapts
to VEN-driven signals during epochs 5--25, after which VEN removal disrupts
the co-adapted weight state.

\begin{figure}[htbp]
\centering
\includegraphics[width=0.90\linewidth]{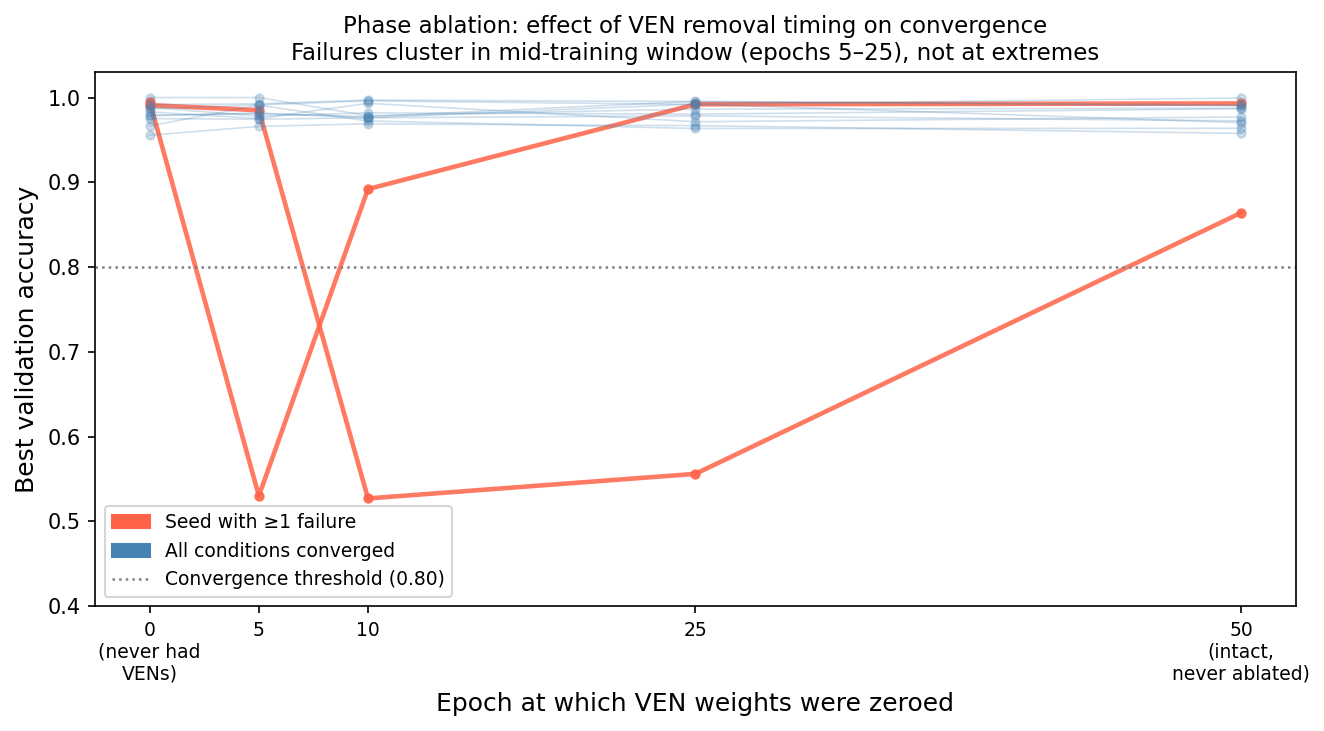}
\caption{Phase ablation results.
Each line shows one seed's best validation accuracy across five ablation
epochs.
Red lines: seeds with at least one failure; blue lines: seeds converging
in all conditions.
Failures cluster in the mid-training window (epochs 5--25); zero failures
occur at the extremes (epoch~0 or epoch~50).}
\label{fig:phase}
\end{figure}

\begin{table}[H]
\centering
\caption{Phase ablation failure counts.
$n$\,=\,seeds with complete data at that ablation epoch.}
\label{tab:phase}
\begin{tabular}{lccc}
\toprule
Ablation epoch & $n$ & Failures & Rate \\
\midrule
0  (start, never had VENs) & 13 & 0 & 0\%  \\
5                          & 13 & 2 & 15\% \\
10                         & 12 & 1 & 8\%  \\
25                         & 12 & 1 & 8\%  \\
50 (no ablation, intact)   & 12 & 0 & 0\%  \\
\bottomrule
\end{tabular}
\end{table}

\FloatBarrier

\subsection{Spectral Norm at Initialisation Is Uniform Across Seeds}
\label{sec:spectral}

Spectral norm measurements across 25 seeds revealed that
$\lVert W_{pp}^{(0)} \rVert_2$ was strikingly uniform: range 0.077--0.079,
mean 0.078 (Figure~\ref{fig:spectral}).
In the VEN-ablated condition, 23 of 25 seeds converged (92\%) and 2 failed;
failed and converged seeds showed identical spectral norm values.
The predicted association between lower initial spectral norm and training
failure was not observed.
The two failures in this subsample correspond to seeds 300078 and 300169,
the same seeds identified as failures in Experiment~2, confirming that
convergence failure is seed-specific and reproducible; the remaining first
25 seeds converged in both experiments.
The null result of interest (no correlation between $\lVert W_{pp}^{(0)} \rVert_2$
and training outcome) holds across both samples.

\begin{figure}[htbp]
\centering
\includegraphics[width=\linewidth]{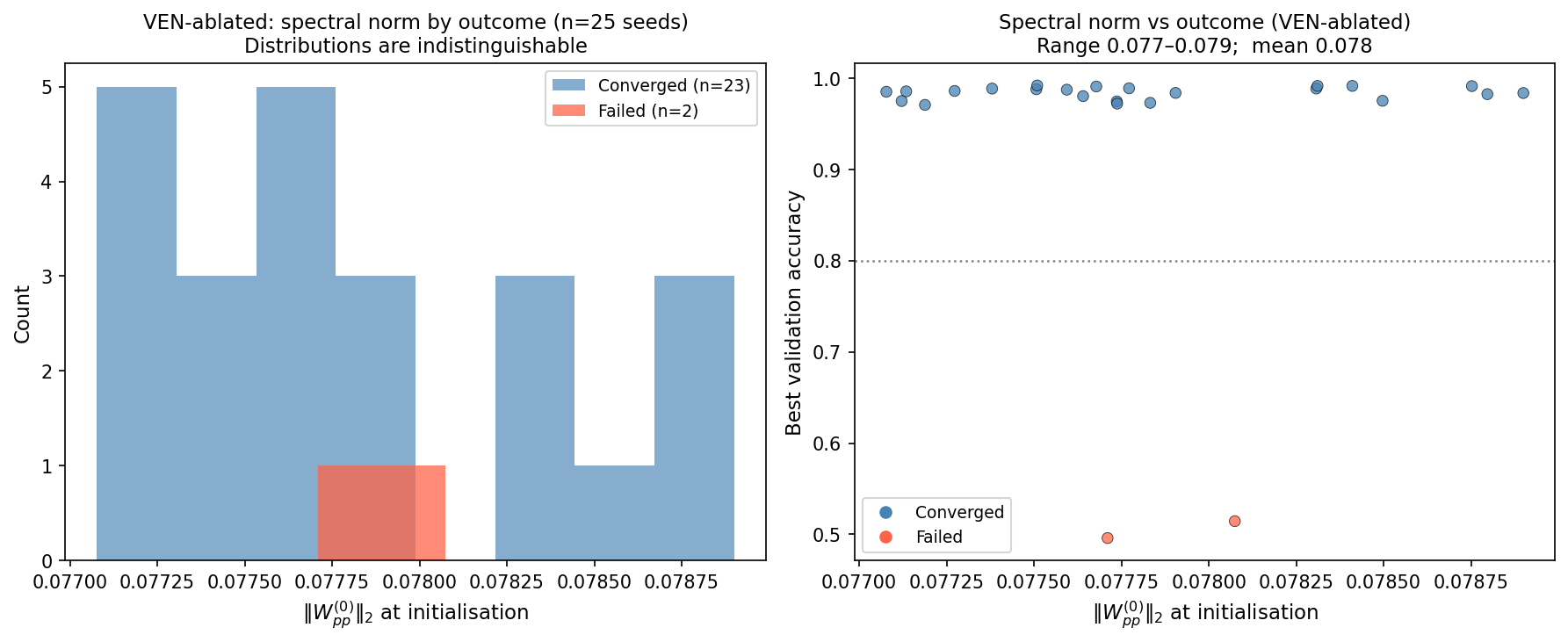}
\caption{Left: distribution of $\lVert W_{pp}^{(0)} \rVert_2$ by training
outcome in the VEN-ablated condition.
Failed seeds (red) and converged seeds (blue) are indistinguishable.
Right: scatter of spectral norm vs best validation accuracy (VEN-ablated).
The uniform spectral norm precludes per-seed prediction of failure.}
\label{fig:spectral}
\end{figure}

This null result is nonetheless theoretically informative.
With membrane leak factor $\beta_{\rm pyr} = 0.95$ and ATan surrogate bound
$\gamma = 1$, the theoretical upper bound on each recurrent Jacobian factor
satisfies:
\begin{equation}
    \alpha = \beta_{\rm pyr} + \gamma\,\lVert W_{pp}^{(0)} \rVert_2
           = 0.95 + 1.0 \times 0.078 = 1.028
    \label{eq:alpha_measured}
\end{equation}
for \emph{every} seed in the experiment.
All networks initialise near the critical gradient-flow boundary ($\alpha \approx 1$),
at the transition between guaranteed attenuation ($\alpha < 1$) and potential
explosion ($\alpha > 1$).
Gradient clipping handles the explosion case but suppresses the gradient
magnitude rather than restoring it; whether the effective learning signal at
$W_{pp}$ is adequate for convergence depends on the instantaneous Jacobian
eigenstructure, not on $\alpha$ alone.
The VEN pathway advantage is therefore not a remedy for unusually poor
initialisations: it is a structural benefit present for all initialisations
because VEN gradients bypass the recurrent Jacobian entirely
(Proposition~\ref{prop:ven_gradient}).

\FloatBarrier

\subsection{Gradient Norms at $W_{pp}$: Unexpected Direction}
\label{sec:gradnorm}

Contrary to the prediction that VEN-ablated networks would show suppressed
gradient magnitude at $W_{pp}$, ablated networks showed \emph{higher}
early-epoch gradient norms than intact networks across all 10 seeds
(ablated mean $= 371{,}871$; intact mean $= 244{,}887$;
ratio $= 1.52\times$;
Wilcoxon signed-rank $W = 12.0$, $p = 0.947$, rank-biserial $r = -0.564$;
Figure~\ref{fig:gradnorms}).

\begin{figure}[htbp]
\centering
\includegraphics[width=\linewidth]{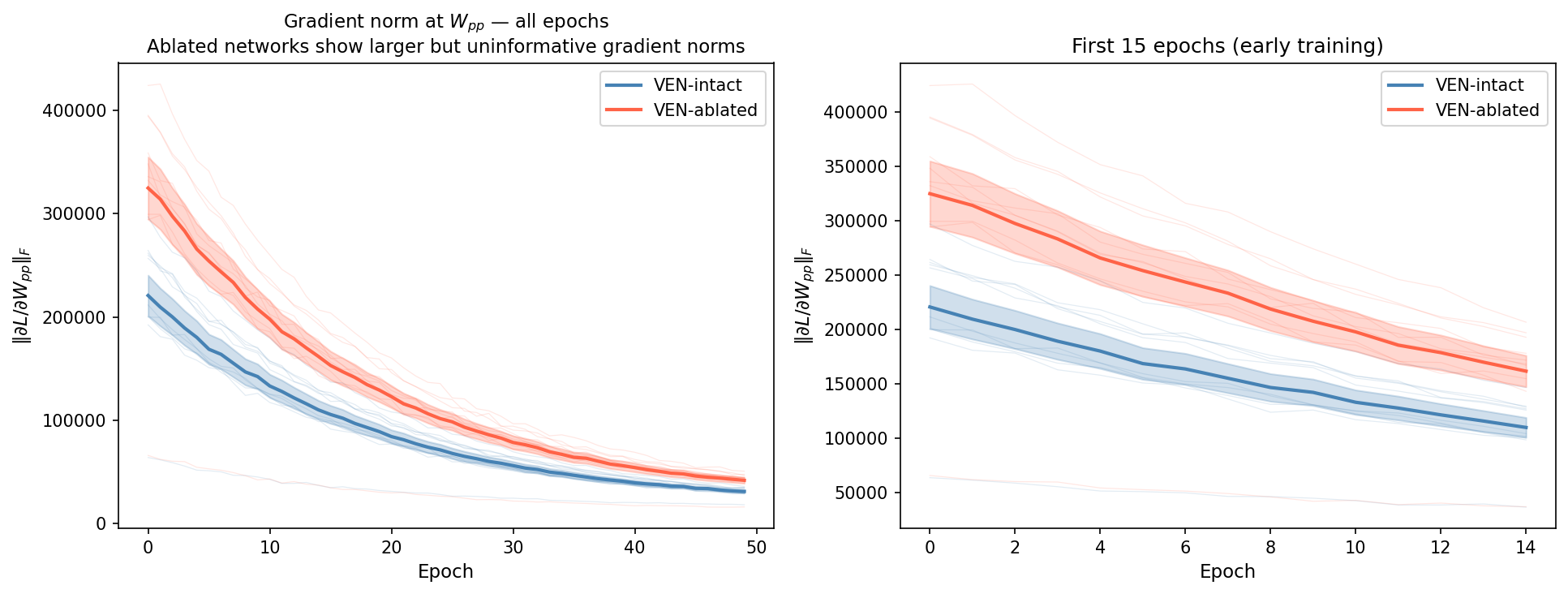}
\caption{Left: mean $\pm$ SEM gradient norm at $W_{pp}$ across all 50
training epochs.
Right: zoom on the first 15 epochs.
VEN-ablated networks (red) show larger, more variable gradient norms than
VEN-intact networks (blue), contrary to a simple gradient-suppression account.}
\label{fig:gradnorms}
\end{figure}

The elevated gradient norm in ablated networks is interpretable in terms of
loss dynamics: without VENs, the loss remains large throughout early training
(the network has no direct output pathway), driving large gradients everywhere,
including at $W_{pp}$.
In intact networks, VENs rapidly reduce the loss via the direct output pathway,
yielding smaller but more informative gradient signals to the pyramidal circuit.
The failed ablated seed (300078) showed early gradient norms nearly identical
to its intact counterpart (ablated: 65{,}427; intact: 63{,}292), confirming
that gradient \emph{magnitude} at $W_{pp}$ does not distinguish seeds that fail
from seeds that converge.

These results indicate that the gradient-attenuation mechanism operates through
factors other than $\lVert \partial L / \partial W_{pp} \rVert_F$ at the
recurrent weight matrix specifically.
The true mechanistic pathway, whether through gradient direction consistency,
loss landscape geometry, or the quality of early error signals reaching
feedforward input weights, is an empirical question for future work.

\FloatBarrier

\subsection{Inference-Time VEN Ablation: Heterogeneous Performance Effects}
\label{sec:ftd}

Across 20 VEN-intact networks trained to convergence, zeroing VEN weights at
test time produced a statistically significant but heterogeneous performance
drop (Wilcoxon signed-rank $W = 15.0$, $p = 0.022$;
Figure~\ref{fig:ftd}).

\begin{figure}[htbp]
\centering
\includegraphics[width=\linewidth]{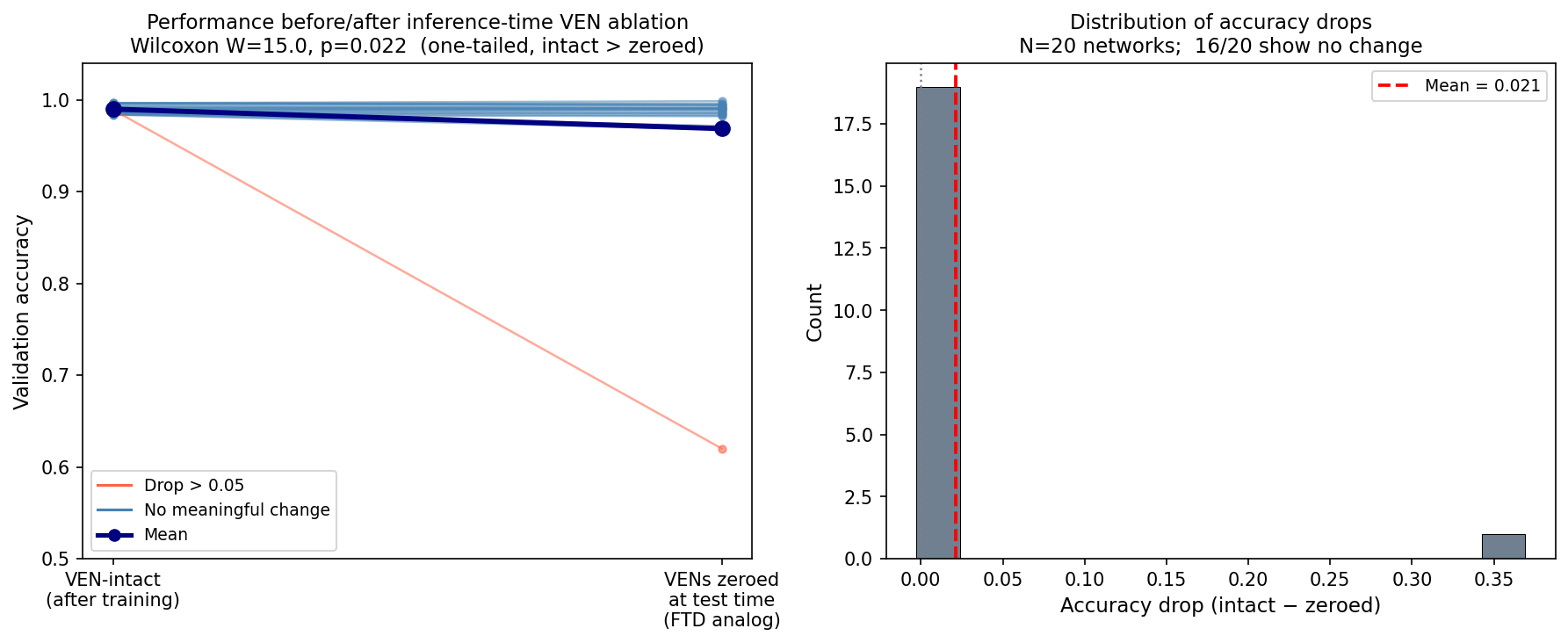}
\caption{Left: paired validation accuracy before (VEN-intact) and after
(VENs zeroed at test time) for 20 networks.
Right: distribution of accuracy drops.
Mean drop $= 0.021 \pm 0.080$.
One network showed catastrophic collapse ($0.989 \to 0.620$, seed 500153);
the majority (16/20) showed no change.}
\label{fig:ftd}
\end{figure}

The effect was highly heterogeneous (Table~\ref{tab:ftd}):
16 of 20 networks (80\%) showed no performance change (drop $\leq 0.002$),
indicating that the learned representation transferred completely to the
pyramidal circuit during training.
Three networks showed minor drops (range 0.007--0.021).
One network exhibited catastrophic collapse: accuracy fell from 0.989 to
0.620, well below the convergence threshold.

\begin{table}[H]
\centering
\caption{Summary of inference-time VEN ablation across 20 networks.}
\label{tab:ftd}
\begin{tabular}{lc}
\toprule
Measure & Value \\
\midrule
Mean accuracy, VEN-intact      & $0.990 \pm 0.012$ \\
Mean accuracy, VENs zeroed     & $0.969 \pm 0.081$ \\
Mean accuracy drop             & $0.021 \pm 0.080$ \\
Networks with drop $\leq 0.002$ & 16/20 (80\%) \\
Networks with drop $> 0.05$    &  1/20 (5\%)  \\
Maximum drop (seed 500153)     & 0.369 \\
Wilcoxon signed-rank           & $W = 15.0$, $p = 0.022$ \\
\bottomrule
\end{tabular}
\end{table}

The result is driven by the consistency of direction (all non-zero drops
favour the intact condition) and one outlier.
We interpret the heterogeneity as reflecting two qualitatively different
training trajectories: networks that develop a VEN-independent pyramidal
representation (80\%) versus networks that co-adapt strongly to VEN signals
and remain VEN-dependent at inference (20\%).
The catastrophic collapse of seed 500153 suggests that, for a minority of
weight initialisations, VEN activity during training becomes load-bearing for
the output representation in a way that the pyramidal circuit cannot compensate.

A model-architectural caveat applies: in the VENCircuit, the same input
$\mathbf{x}$ is available to pyramidal neurons via $W_{ip}$ independently of
VEN status, making the VEN and pyramidal input streams informationally
redundant.
In biological cortex, VENs in ACC layer~V project to subcortical structures
including the striatum, amygdala, and hypothalamus~\citep{allman2011voneconomo}
that do not receive equivalent direct pyramidal input.
The fraction of networks showing VEN-dependence at inference (20\%) may
therefore underestimate the biological consequence of adult VEN loss.

\FloatBarrier

\section{Theoretical Analysis: Von Economo Neurons as Residual Gradient Pathways}
\label{sec:theory}

\subsection{Network Specification and Surrogate-Gradient BPTT}
\label{sec:theory:model}

Let $\mathbf{u}_t \in \mathbb{R}^N$ and $\mathbf{s}_t \in \{0,1\}^N$ denote
pyramidal membrane potentials and spike vectors at timestep $t$.
Ignoring the detached reset for gradient analysis, the gradient-relevant
pyramidal recurrence is:
\begin{align}
    \mathbf{u}_t &= \beta_{\rm pyr}\,\mathbf{u}_{t-1}
                   + W_{pp}\,\mathbf{s}_{t-1}
                   + W_{ip}\,\mathbf{x},
                   \label{eq:pyr_dyn}\\
    \mathbf{s}_t &= \sigma_{\rm spike}(\mathbf{u}_t - \theta),
                   \label{eq:pyr_spike}
\end{align}
where $\beta_{\rm pyr} = 0.95$ and $\sigma_{\rm spike}$ uses the ATan
surrogate~\eqref{eq:atan_surr}.
VEN dynamics are feedforward-only (no recurrent pyramidal input):
\begin{equation}
    \mathbf{z}_t = \beta_{\rm VEN}\,\mathbf{z}_{t-1} + W_{iv}\,\mathbf{x},
    \qquad
    \mathbf{v}_t = \sigma_{\rm spike}(\mathbf{z}_t - \theta),
    \label{eq:ven_dyn}
\end{equation}
where $\beta_{\rm VEN} = 0.80$.
Training uses surrogate-gradient BPTT~\citep{werbos1990bptt,
neftci2019surrogate, zenke2018superspike}, replacing $\sigma_{\rm spike}'$
with the ATan derivative in~\eqref{eq:atan_surr}.
The ATan surrogate is bounded: $|\sigma'(u)| \leq \gamma = 1$ for all $u$
(maximum at $u = 0$, decaying as $u^2$).
Define $D_t = \mathrm{diag}\!\left(\sigma'(\mathbf{u}_t - \theta)\right)
\in \mathbb{R}^{N \times N}$, so $\lVert D_t \rVert_2 \leq \gamma = 1$.

\subsection{Gradient Attenuation in the Pyramidal Pathway}
\label{sec:theory:attenuation}

\begin{proposition}[Recurrent Jacobian bound]
\label{prop:attenuation}
Define $\alpha = \beta_{\rm pyr} + \gamma\lVert W_{pp} \rVert_2$.
For all $k > t$, the Jacobian of $\mathbf{u}_k$ with respect to
$\mathbf{u}_t$ satisfies:
\begin{equation}
    \left\lVert\frac{\partial \mathbf{u}_k}{\partial \mathbf{u}_t}
    \right\rVert_2 \leq \alpha^{k-t}.
    \label{eq:jac_bound}
\end{equation}
\end{proposition}

\begin{proof}
The Jacobian of $\mathbf{u}_k$ w.r.t.\ $\mathbf{u}_t$ is
$\prod_{j=t+1}^{k}J_j$ where $J_j = \beta_{\rm pyr} I + D_{j-1}
W_{pp}^{\top}$.
By submultiplicativity,
$\lVert J_j \rVert_2 \leq \beta_{\rm pyr} + \lVert D_{j-1} \rVert_2
\lVert W_{pp} \rVert_2 \leq \beta_{\rm pyr} + \gamma\lVert W_{pp} \rVert_2
= \alpha$.
Inequality~\eqref{eq:jac_bound} follows by taking the product over
$j = t+1, \ldots, k$.
\end{proof}

\begin{remark}
At random initialisation, spectral norm measurements
(Section~\ref{sec:spectral}) give $\lVert W_{pp}^{(0)} \rVert_2 \approx 0.078$
uniformly across all 25 measured seeds.
With $\beta_{\rm pyr} = 0.95$ and the ATan bound $\gamma = 1$, the
theoretical bound on each Jacobian factor is
$\alpha = 0.95 + 0.078 = 1.028$, placing every network at the near-critical
point of gradient-flow stability.
When $\alpha < 1$, the product in~\eqref{eq:jac_bound} decays exponentially
(gradient vanishing); when $\alpha > 1$, it can grow (gradient explosion).
At $\alpha \approx 1$, neither is guaranteed by this bound, and the actual
behaviour depends on the instantaneous neural activity pattern (specifically,
on $\lVert D_t \rVert_2$, which is far below $\gamma = 1$ for neurons whose
membrane potential is distant from threshold).
Gradient clipping addresses explosion but suppresses gradient magnitude
rather than restoring it; it does not remedy effective vanishing.
The instability of recurrent networks near the critical spectral radius is a
classical result in dynamical systems theory~\citep{sompolinsky1988} and was
formally characterised as the vanishing gradient problem by
\citet{bengio1994learning}, motivating the LSTM architecture~\citep{hochreiter1997lstm}.
\end{remark}

\subsection{The VEN Pathway Provides an $O(1)$ Gradient}
\label{sec:theory:ven_gradient}

\begin{proposition}[VEN direct gradient]
\label{prop:ven_gradient}
Under the approximate readout~\eqref{eq:output} and treating
$\partial\mathbf{z}_t/\partial W_{iv} = \mathbf{x}^{\top}$ at each step
(one-step gradient; see Remark~\ref{rem:ven_bptt}), the gradient of $L$
with respect to $W_{iv}$ satisfies:
\begin{equation}
    \frac{\partial L}{\partial W_{iv}}
    = \frac{1}{T}\sum_{t=1}^{T}
      \Bigl(W_{vo}^{\top}\,\tfrac{\partial L}{\partial \hat{\mathbf{y}}}\Bigr)\,
      \sigma'\!\left(\mathbf{z}_t - \theta\right)\,\mathbf{x}^{\top},
    \label{eq:ven_grad}
\end{equation}
and analogously for $W_{vo}$.
This gradient contains no product of recurrent Jacobians of the
form~\eqref{eq:jac_bound}: it is $O(1)$ with respect to $T$
and \emph{independent of} $\lVert W_{pp} \rVert_2$.
\end{proposition}

\begin{proof}
From~\eqref{eq:ven_dyn}, $\mathbf{z}_t$ depends on $W_{iv}$ but not on
$W_{pp}$ or $\mathbf{s}_t$.
Under the one-step gradient approximation
($\partial\mathbf{z}_t/\partial W_{iv} = \mathbf{x}^{\top}$),
differentiating the approximate readout~\eqref{eq:output} through the VEN
pathway gives~\eqref{eq:ven_grad} directly, with no factor of the form
$(\beta I + D_j W_{pp}^{\top})$.
Since $\lvert\sigma'(\cdot)\rvert \leq \gamma = 1$ pointwise, each of the
$T$ terms is bounded in magnitude by
$\lVert W_{vo}^{\top}\,\partial L/\partial\hat{\mathbf{y}}\rVert_2
\lVert\mathbf{x}\rVert_2$.
The prefactor $1/T$ cancels the sum of $T$ such terms, giving the
$O(1)$ bound
$\gamma\,\lVert W_{vo}^{\top}\,\partial L/\partial\hat{\mathbf{y}}\rVert_2
\lVert\mathbf{x}\rVert_2$.
\end{proof}

\begin{remark}
\label{rem:ven_bptt}
The full surrogate-gradient BPTT through the VEN's own $\beta_{\rm VEN}$
recurrence multiplies the factor $\mathbf{x}^{\top}$ in each term of
\eqref{eq:ven_grad} by a scalar $c_t \geq 1$.
The scalar satisfies the recursion
$c_t = (1-s_{t-1}^{\rm det})\,\beta_{\rm VEN}\,c_{t-1} + 1$ (with $c_1=1$, since $\mathbf{z}_0=0$),
so $c_t \leq (1-\beta_{\rm VEN})^{-1} = 5$ for all $t$ and all spike
patterns.
This multiplies the $O(1)$ bound by at most~5, leaving the bound
independent of~$T$.
The absence of $W_{pp}$ from the gradient holds exactly, without
approximation.
\end{remark}

\subsection{Corollary: Training Stability Asymmetry}
\label{sec:theory:corollary}

\begin{corollary}[Training stability asymmetry]
\label{cor:stability}
\leavevmode
\begin{enumerate}
    \item \textbf{VEN-ablated network.}\\
    The gradient reaching all parameters flows exclusively through the
    near-critical recurrent pathway
    ($\alpha \approx 1.028$, Section~\ref{sec:spectral}).
    Whether any given initialisation converges depends on the instantaneous
    Jacobian eigenstructure, not on the spectral norm at initialisation alone;
    approximately 30\% of seeds fail.

    \item \textbf{VEN-intact network.}\\
    Regardless of $\alpha$, the VEN pathway provides an $O(1)$ gradient
    to $W_{iv}$ and $W_{vo}$ (Proposition~\ref{prop:ven_gradient}).
    A partial solution can be learnt through the VEN pathway even when the
    recurrent Jacobian is near-critical, providing a learning scaffold that
    prevents catastrophic gradient failure.
\end{enumerate}
Since $\alpha \approx 1.028$ universally at initialisation, VEN-intact
networks are protected from gradient instability by a structural architectural
property, not by a favourable initialisation.
VEN-ablated networks that converge do so despite, not because of, a
beneficial spectral norm.
\end{corollary}

\subsection{Relationship to Residual Networks}
\label{sec:theory:resnet}

The VEN direct gradient pathway is mathematically analogous to the skip
connections of residual networks~\citep{he2016deep}: a route through which
gradients propagate without passing through the ill-conditioned intermediate
computation.
In ResNets, skip connections bypass depth; in the VENCircuit, they bypass
temporal recurrence.
Proposition~\ref{prop:ven_gradient} formalises this analogy: the VEN pathway
is a biological residual connection whose gradient is structurally bounded
away from the instabilities of the recurrent Jacobian, regardless of the
operating point.

\subsection{Note on Empirical Validation}
\label{sec:theory:validation}

Gradient norm measurements at $W_{pp}$ (Section~\ref{sec:gradnorm}) yielded
a result opposite to the simple prediction: ablated networks showed
\emph{larger} gradient norms than intact networks.
We interpret this as reflecting loss dynamics rather than gradient pathway
capacity: high loss in early ablated training drives large gradient magnitudes
everywhere, including at $W_{pp}$, but these gradients are not informative
enough to produce consistent weight updates.
The mathematical account in Propositions~\ref{prop:attenuation} and
\ref{prop:ven_gradient} remains valid as a structural analysis of pathway
capacity; the empirical result indicates that gradient \emph{magnitude} at
$W_{pp}$ is not the operative readout of the mechanism.
We therefore interpret the VEN learning advantage as likely operating through
gradient \emph{direction consistency} across batches, or through a smoother
loss landscape accessible when the VEN pathway carries early error signal,
rather than through raw gradient magnitude.
This reinterpretation is consistent with the scaffold hypothesis (Section~\ref{sec:discussion})
and is a direct priority for future experimental work; measuring gradient
cosine similarity across mini-batches, or gradient signal at $W_{ip}$,
would constitute a targeted test.

\FloatBarrier

\section{Clinical Predictions: bvFTD and ASC}
\label{sec:clinical}

The quantitative figures cited below (e.g., 30\% failure rate, 20\% inference
sensitivity) are specific to the VENCircuit architecture and training regime.
They should be interpreted as qualitative predictions about the direction and
pattern of effects, not as precise quantitative forecasts for clinical
populations.

\paragraph{Developmental VEN reduction (ASC analog).}
The VEN-ablated-from-birth condition produces unreliable acquisition:
approximately 30\% of networks never learn the task, while 70\% converge to
high accuracy.
Converging networks are indistinguishable from intact networks in final
performance.
This maps onto the clinically observed variability in social skill acquisition
in ASC~\citep{baroncohen1985}: some individuals acquire robust social cognitive
abilities, others do not, and those who acquire them can perform competently
in structured contexts.
The model predicts this variability is not random noise but reflects whether
the recurrent pyramidal circuit at a critical developmental window has
sufficient architecture to learn without VEN scaffolding.

\paragraph{Adult VEN loss (bvFTD analog).}
Inference-time VEN ablation in trained networks produced significant
performance drops (Wilcoxon $p = 0.022$), but with marked heterogeneity:
80\% of networks were unaffected, while 20\% showed meaningful drops including
one catastrophic collapse.
This heterogeneity is itself a prediction: not all individuals with bvFTD
will show equivalent severity of social cognitive decline at the same stage
of VEN loss, depending on how strongly their acquired representations depend
on ongoing VEN activity.
The fraction showing severe effects (20\% in this model) is likely an
underestimate of the biological case due to the architectural redundancy
noted in Section~\ref{sec:ftd}.

\paragraph{The acquisition asymmetry.}
Both conditions share a common mechanism: VENs are a training scaffold.
Their developmental absence makes acquisition unreliable; their adult loss
disrupts representations that depended on them.
The key distinction is timing relative to acquisition:
\begin{itemize}
    \item \emph{Before acquisition} (ASC analog): the scaffold is absent from
    the start; some initialisations fail entirely, others succeed without it.
    \item \emph{After acquisition} (bvFTD analog): the scaffold is removed;
    networks that developed VEN-dependent representations partially or fully
    lose the ability to perform the task.
\end{itemize}
This asymmetry is not predicted by accounts in which VENs act as a static
processing gain; it requires that VENs serve a dynamic role during the
weight-adaptation process.

\section{Discussion}
\label{sec:discussion}

We begin with an explicit statement of scope.
The VENCircuit is a deliberately simplified model: its ``VEN-like'' neurons
are defined by three architectural properties (feedforward-only input,
direct output projection, faster time constant) rather than by any attempt
to recapitulate the full morphological or connectivity profile of biological
VENs.
The convergence advantage we report is a fact about this architecture;
whether the same mechanism operates in biological cortex is a hypothesis
the model motivates, not a result it establishes.
The clinical framing throughout should be read in that spirit: as a set
of model-derived predictions whose value lies in being falsifiable, not as
claims about human neuropathology.

\subsection{Computational Analogy: VEN-like Pathways as Residual Gradient Channels}

The present results support a functional analogy between VENs and
residual connections in deep networks~\citep{he2016deep}.
VENs project directly from input to output-adjacent structures, bypassing
the recurrent pyramidal dynamics.
Proposition~\ref{prop:ven_gradient} shows that this pathway carries a
gradient that is $O(1)$ regardless of the recurrent weight configuration,
while Proposition~\ref{prop:attenuation} shows that the purely recurrent
pathway operates near the critical gradient-flow boundary ($\alpha \approx 1.028$
at initialisation).
The analogy is motivated, though not confirmed, by VEN morphology:
thick, myelinated apical dendrites; sparse local connectivity; long-range
subcortical projections~\citep{allman2011voneconomo} are consistent with
a direct, low-recurrence signalling channel.
Whether this morphology produces a functionally analogous gradient advantage
in cortex is an empirical question the model cannot answer.

\subsection{Acquisition Scaffolds: Why Dispensability After Learning Is a Predicted, Not Post Hoc, Result}

The scaffold hypothesis makes a specific, non-obvious prediction before any
experiment is run: if VENs function during learning to guide weight
configuration rather than to permanently carry the computation, then a
network trained with VENs should be largely able to perform the task
without them.
This is precisely what Experiments 1 and 6 show jointly, and it is why
the combination of results is evidence for the scaffold account rather than
a convenient reconciliation of inconsistent findings.
VENs are largely dispensable
for performance once convergence is achieved.
Successful VEN-ablated networks match intact networks in final accuracy, and
80\% of VEN-intact networks retain full performance when VENs are removed at
test time.
This is not a null result: it tells us that the pyramidal recurrent circuit
is \emph{sufficient} for the computation, but reliably arriving at the
appropriate weight configuration requires the VEN scaffold during learning.

This framing has a specific developmental prediction: the social cognitive
consequences of VEN loss are determined more by its timing relative to the
acquisition window than by the loss itself.
Developmental VEN reduction (as in ASC) impairs the reliability of entering
the learning process.
Adult VEN loss (as in bvFTD) primarily affects networks that co-adapted
strongly to VEN signals; the majority of acquired representations survive.

\subsection{Mechanistic Validation: Honest Assessment}

Two mechanistic predictions of the gradient-pathway account were not confirmed
in the expected direction.
The spectral norm at initialisation was uniform across all 25 measured seeds
($\lVert W_{pp}^{(0)} \rVert_2 \approx 0.078$ for every seed), making
per-seed prediction of failure impossible from this measure.
Gradient norms at $W_{pp}$ were higher in ablated networks, the opposite of
the prediction.
These results indicate that the mechanism does not operate through gradient
magnitude at $W_{pp}$ alone.
We interpret both results as consistent with the broader framework: the
uniform spectral norm confirms universal near-critical dynamics (strengthening
the case for VEN-independent gradient attenuation risk), while the elevated
ablated gradient norms reflect uninformative high-loss gradients rather than
a functioning learning signal.
Future work should measure gradient direction consistency and gradient magnitude
at $W_{ip}$ to test whether the mechanism is visible at those readouts.

\subsection{Relationship to Prior Computational Models}

Previous computational accounts of VEN function have focused on rapid
affective signalling~\citep{allman2005intuition}.
We are not aware of any prior computational model that addresses the role
of VENs in the \emph{learning process} itself.
The present work connects to the literature on credit assignment in biological
networks~\citep{lillicrap2020backpropagation, sacramento2018} and to the
specific challenge of training recurrent spiking
circuits~\citep{bellec2020, zenke2018superspike}.
VENs may represent a biological structural solution to the credit assignment
problem: a direct projection pathway through which error signals propagate
without recurrent attenuation, analogous to the eligibility-propagation
mechanisms proposed in~\citet{bellec2020} but implemented through anatomy
rather than learning rules.

\subsection{Limitations and Future Work}

Several caveats constrain interpretation.
First, and most importantly, VENs in the VENCircuit receive the same input
$\mathbf{x}$ as pyramidal neurons; a bypass-only control (replacing
VEN-to-output connections with direct input-to-output connections of matched
weight count, without any VEN-specific temporal dynamics) was not run.
This means the observed convergence advantage cannot be unambiguously
attributed to the temporal structure of VEN activity rather than to the
presence of an additional direct pathway per se.
This is the single most important follow-up experiment for establishing
whether VEN-specific properties, rather than bypass topology alone, drive
the convergence advantage reported here.
Second, the social classification task uses synthetic burst-modulated Poisson
spike trains; real social cognition involves temporal, multimodal, and
contextual complexity not captured here.
Third, the LIF dynamics use the detached-reset approximation; exact gradient
bounds for hard-reset LIF dynamics may differ slightly from Proposition~1.
Fourth, the mechanistic validation experiments gave unexpected results that
do not straightforwardly confirm the gradient magnitude account, as discussed
above.
Fifth, gradient norm tracking was conducted on 10 seeds (Experiment~5), and
inference-time ablation on 20 seeds (Experiment~6); these sample sizes are
modest and the observed heterogeneity (one catastrophic collapse in 20 seeds)
should be interpreted with appropriate caution pending replication in larger
samples.

Future work should: (i) measure gradient direction consistency and
feedforward input-weight gradient norms; (ii) test the clinical predictions
using patient-derived iPSC organoids in which VEN density can be modulated;
(iii) model the non-redundant subcortical projection pathway of biological
VENs to better capture the FTD analog; and (iv) extend the architecture to
hierarchical cortical circuits to test whether the acquisition scaffold role
generalises across tasks.

\section{Conclusion}
\label{sec:conclusion}

We have shown that a biologically motivated spiking neural network model
equipped with Von Economo neuron-like projection neurons is 21 times more
likely to successfully acquire a social cognitive representation than an
otherwise identical network without VENs
(OR $= 21.0$, 95\% CI 2.7--167, Fisher's exact $p = 8.7 \times 10^{-5}$).
This advantage does not reflect increased network capacity but a specific
architectural property: VENs provide a direct pathway that bypasses the
recurrent Jacobian instabilities of the pyramidal circuit, offering a
learning scaffold that is mathematically derivable and empirically consequential.
Empirical gradient measurements indicate that the operative mechanism is
likely gradient direction consistency rather than magnitude, a distinction
that constitutes the primary target for follow-on mechanistic work.
Inference-time VEN ablation produces statistically significant but
heterogeneous performance drops, consistent with a subset of networks
developing VEN-dependent output representations during training.
The resulting framework predicts that the social cognitive consequences of
VEN loss are determined by its timing relative to the developmental acquisition
window, generating specific and falsifiable hypotheses for neuroimaging,
organoid, and electrophysiology research.

\section*{Data and Code Availability}

All experiment code and result files are available at:
\url{https://github.com/esila-keskin/VENCircuit}

\appendix

\section{Hyperparameters}
\label{app:hyperparams}

\begin{table}[htbp]
\centering
\caption{VENCircuit training hyperparameters for all main experiments
(Experiments 1--6). All experiments use the same configuration unless
otherwise noted.}
\label{tab:hyperparams}
\renewcommand{\arraystretch}{1.15}
\begin{tabular}{llll}
\toprule
Parameter & Symbol & Value & Note \\
\midrule
\multicolumn{4}{l}{\textit{Architecture}} \\
Pyramidal neurons          & $N$              & 2{,}000   & \\
VEN neurons (intact)       & $K$              & 40        & $= 0.02N$ \\
Pyramidal fan-in           &                  & 80        & sparse $W_{ip}$ \\
Pyramidal recurrent prob.  & $p_{\rm rec}$   & 0.15      & sparse $W_{pp}$ \\
VEN fan-in                 &                  & 8         & sparse $W_{iv}$ \\
Input dimension            & $d$              & 100       & input neurons \\
\midrule
\multicolumn{4}{l}{\textit{LIF dynamics}} \\
Pyramidal time constant    & $\tau_{\rm pyr}$ & 20\,ms    & $\beta_{\rm pyr}=0.95$ \\
VEN time constant          & $\tau_{\rm VEN}$ & 5\,ms     & $\beta_{\rm VEN}=0.80$ \\
Pyramidal / VEN threshold  & $\theta$         & 0.50      & \\
Output threshold           & $\theta_{\rm out}$& 0.10     & WTA readout \\
Surrogate gradient         &                  & ATan      & Eq.~\eqref{eq:atan_surr}\\
\midrule
\multicolumn{4}{l}{\textit{Task}} \\
Simulation timesteps       & $T$              & 50        & at 1\,ms resolution \\
Training samples           &                  & 6{,}000   & \\
Validation samples         &                  & 1{,}000   & \\
Test samples               &                  & 2{,}000   & \\
Batch size                 &                  & 64        & \\
\midrule
\multicolumn{4}{l}{\textit{Training}} \\
Training epochs            &                  & 50        & \\
Learning rate              & $\eta$           & $1\times10^{-3}$ & Adam initial LR \\
LR scheduler               &                  & CosineAnnealingLR & $T_{\max}=50$ epochs \\
Weight decay               & $\lambda$        & $1\times10^{-5}$ & $L_2$ \\
Gradient clip              &                  & 1.0       & global norm \\
Optimiser                  &                  & Adam      & \citep{kingma2015adam} \\
Convergence criterion      &                  & val $\geq 0.80$ & \\
\midrule
\multicolumn{4}{l}{\textit{Weight initialisation ($\sigma = 0.1/\sqrt{N}$)}} \\
$W_{ip}$ (pyramidal input) &                  & $\mathcal{N}(0,\sigma)$ & sparse \\
$W_{pp}$ (pyramidal recur.)&                  & $\mathcal{N}(0,\sigma)$ & sparse \\
$W_{po}$ (pyr.\ output)    &                  & $\mathcal{N}(0,0.05)$   & \\
$W_{iv}$ (VEN input)       &                  & $\mathcal{N}(0,2\sigma)$ & sparse \\
$W_{vo}$ (VEN output)      &                  & $\mathcal{N}(0,0.03)$   & \\
\bottomrule
\end{tabular}
\end{table}

\FloatBarrier

\section{STDP Credit Assignment: Preliminary Negative Result}
\label{app:stdp}

Prior to the training-reliability analysis, we tested whether VEN synapses
generate stronger per-synapse reward-modulated STDP credit signals than
pyramidal synapses.
Networks were trained in two phases: (1) full backpropagation training for
20 epochs to establish input representations, then (2) output
weights $W_{po}$ and $W_{vo}$ were reinitialised and trained via
reward-modulated STDP while input and recurrent weights were frozen.
Across five VEN-intact seeds and five VEN-ablated seeds, STDP learning
on $W_{vo}$ and $W_{po}$ produced chance-level final accuracy in both
conditions (VEN-intact: $0.40 \pm 0.16$; VEN-ablated: $0.48 \pm 0.03$;
$p = 0.38$, two-sample $t$-test), with no difference in learning speed.
The hypothesis that VENs accelerate output credit assignment via STDP
was not confirmed.
This null result redirected the investigation toward training reliability
under standard backpropagation, which revealed the convergence asymmetry
reported in Section~\ref{sec:results:fisher}.

\bibliographystyle{abbrvnat}
\bibliography{references}

\end{document}